\documentclass[11pt]{article}
\usepackage{amsmath,amsfonts, amsthm, amssymb, mathrsfs}
\usepackage{verbatim}
\usepackage[usenames,dvipsnames,table]{xcolor}
\usepackage{hyperref}
\usepackage{fullpage}

\usepackage{graphicx}
\graphicspath{{figs/}{}}

\providecommand{\abs}[1]{\lvert#1\rvert}
\DeclareMathOperator{\supp}{supp}

\theoremstyle{definition}
    \newtheorem{thm}{Theorem}[section]
    \newtheorem{lemma}[thm]{Lemma}
    \newtheorem{cor}[thm]{Corollary}
    \newtheorem{defn}[thm]{Definition}
    \newtheorem*{rethm}{Theorem}
    \newtheorem*{redefn}{Definition}

\theoremstyle{remark}
    \newtheorem*{rmk}{Remark}
    \newtheorem*{notation}{Notation}

\title{On the capacity of the binary adversarial wiretap channel}

\author{Carol Wang\thanks{Department of Electrical and Computer Engineering, National University of Singapore. Email: {\tt wangc@nus.edu.sg}. This work was supported by a Ministry of Education Tier 2 Grant (R-263-000-B61-112).}}

\date{}

\begin{document}
\maketitle


\begin{abstract}
New bounds on the semantic secrecy capacity of the binary adversarial wiretap channel are established . Against an adversary which reads a $\rho_r$ fraction of the transmitted codeword and modifies a $\rho_w$ fraction of the codeword, we show an achievable rate of $1-h(\rho_w)-\rho_r$, where $h(\cdot)$ is the binary entropy function. We also give an upper bound which is nearly matching when $\rho_r$ is small. 

\end{abstract}

\section{Introduction}

In the most basic model of communication, a sender attempts to communicate with a receiver over a noisy channel, with the goal of achieving {\em reliability}: the receiver should be able to recover the intended message even in the presence of noise. The wiretap channel, introduced in~\cite{wyner,bcc}, adds a wiretapper or eavesdropper to the model. The wiretapper also has access to a noisy version of the sender's transmission, and now the sender also wants to achieve {\em secrecy}: the wiretapper should not learn anything about the intended message. (These notions will be made precise later.) Typically, the sender is connected to both the receiver and wiretapper by a memoryless broadcast channel, and the dual goals of reliability and secrecy can be met with positive communication rate when the channel to the wiretapper is ``noisier'' than that to the receiver. 

In this work, we consider the extension of the wiretap model to the adversarial setting. We not only allow the wiretapper to choose an arbitrary $\rho_r$ fraction of transmitted symbols to read (this is the ``Wiretap Channel II'' model of Ozarow and Wyner~(\cite{wt2})), but also allow the wiretapper to choose an arbitrary $\rho_w$ fraction of errors to add to the transmission before it reaches the receiver. We refer to this model as the $(\rho_r,\rho_w)$ adversarial wiretap channel. An incomplete survey of related models and results appears in Section~\ref{sec:prev}; this particular model and the name ``adversarial wiretap channel'' were introduced in~\cite{WSN} along with an explicit construction for large alphabets. The authors of~\cite{WSN} also show that the capacity of such a channel is at most $1-\rho_w-\rho_r$. 

The model of the adversarial wiretap channel represents a natural middle ground between truly adversarial errors (i.e.\ $\rho_r=1$, when the adversary has full knowledge of what is being transmitted) and fully oblivious errors~(i.e.\ $\rho_r=0$, defined in~\cite{obliv}), where the adversary has no knowledge of which codeword is being transmitted, but can add arbitrary errors to the codeword. In that respect it is related to the more abstract model of $\gamma$-oblivious channels due to~\cite{obliv}, where, loosely, the parameter $\gamma$ controls how much the channel knows about the transmitted codeword (see also the discussion in Section~\ref{sec:prev}). 
\medskip

The main theorem of this work is the following, which bounds the rate of a binary code which simultaneously achieves reliability and secrecy in the adversarial wiretap model. 

\begin{rethm}[Theorem~\ref{thm:main}] The capacity of the $(\rho_r,\rho_w)$ adversarial wiretap channel is at least $1-h(\rho_w) - \rho_r$. 
\end{rethm}

Loosely, the loss of $\rho_r$ in the achievable rate is required to achieve secrecy, and can be matched in constructions by adding some pseudorandom noise to the transmitted symbols. The loss of $h(\rho_w)$ in the rate then corresponds to what is necessary to correct a $\rho_w$ fraction of errors. The challenge faced by previous work such as~\cite{wt2adv} is that the best known rate for correcting an {\em arbitrary} $\rho_w$ error fraction is $1-h(2\rho_w)$ for binary codes. However, there are a few cases when we can achieve a rate of $1-h(\rho_w)$ against a $\rho_w$ fraction of errors, most notably in the case of {\em random} errors (\cite{shannon}), and for oblivious errors (\cite{obliv}). 

The main point of our achievability analysis, which uses random coding, is that the adversary's errors must behave like random or oblivious errors, even with the auxiliary knowledge of a $\rho_r$ fraction of the codeword. This allows us to show that correcting such errors only requires a $h(\rho_w)$ loss in the achievable rate, allowing for a final rate of $1-h(\rho_w)-\rho_r$. This recalls the work of~\cite{myopic} (``Sufficiently myopic adversaries are blind''), and indeed we build on their techniques to show that the adversary is still ``blind'' after reading his choice of symbols. 

Our result improves on previous known bounds for binary codes, and comes close to matching the upper bound on achievable rate induced by the random wiretap channel (Theorem~\ref{thm:rand-cap}) when $\rho_r$ is small. 

\subsubsection*{Organization}

In Section~\ref{sec:prelim}, we define some terminology relating to the adversarial wiretap model, including the capacity, and build on this to discuss related work.  In Section~\ref{sec:ub}, we show an upper bound on the achievable rate of any family of codes for the adversarial wiretap channel. We also give a lower bound in Section~\ref{sec:lb} using a random code construction, which we show achieves both reliability and secrecy. We conclude in Section~\ref{sec:conc} with some discussion of future work. 

\section{Preliminaries}
\label{sec:prelim}

\subsection{Setup and notation}
\label{sec:ntn}

\begin{notation} We will use the following conventions throughout. 

Unless otherwise noted, all logs are to base $2$. $[n]$ denotes the set $\{1,2,\dotsc, n\}$.  The function $h(\cdot)$ is the binary entropy function $h(p) = -p\log p - (1-p)\log(1-p)$, defined for $p\in[0,1]$. More generally, for a random variable $\mathbf{X}$, we denote by $H(\mathbf{X})$ the entropy of $\mathbf{X}$. 

A binary {\em code} $C$ is a subset of $\{0,1\}^n$ for some integer $n$, the {\em block length}. The {\em rate} of a code $C$ is $R(C):=\log\abs{C}/n$. All references to codes ``of rate $R$'' implicitly mean codes of rate {\em at least} $R$; in particular, if $R<0$, we will assume the code is empty. 

In this work, all codes $C$ come equipped with an arbitrary encoding function which is a bijection between $[\abs{C}]$ and $C$ mapping a {\em message} $m$ to its {\em encoding} $x$. $C$ also admits a decoding function from $\{0,1\}^n\to [\abs{C}]$, the ``nearest neighbor decoder,'' which maps a string $y$ to the message $m$ whose encoding is closest to $y$ in Hamming distance, with ties broken arbitrarily. 

In defining channel capacities, we will think of codes as belonging to a {\em family} of codes (denoted $\mathcal{C}$). A family of codes is a sequence $(C_n\subseteq\{0,1\}^n)_{n\to\infty}$ of codes with growing block length $n$. We are interested in the behavior of the codes in the family as $n\to\infty$. For example, the {\em rate} of a family $\mathcal{C}$ is 
\[R(\mathcal{C}) := \liminf_{n\to\infty} R(C_n).\]
\end{notation}

The focus of this work is on determining the capacity of the {\em binary} adversarial wiretap channel. Although we believe that analogous results hold over larger alphabets, we have not checked this formally. 

\begin{defn} Let $\rho_r\in(0,1)$ and let $\rho_w\in(0,1/2)$. 
The {\bf $(\rho_r,\rho_w)$ adversarial wiretap channel (AWTC)} allows the adversary to read an arbitrary $\rho_r$ fraction of the transmitted codeword and introduce a $\rho_w$ fraction of errors  whose distribution depends only on the code and the symbols read by the adversary. 
\end{defn}

\begin{rmk} This is the model proposed by~\cite{WSN}, where it is also called an AWTC, and is the same as the ``active eavesdropper'' model of~\cite{wt2adv} when $\rho_r=\rho_w$. The adversary may use randomness both in choosing the locations to read and in introducing errors. 

In this definition, we restrict $\rho_r,\rho_w$ to be nonzero. As we will outline in Section~\ref{sec:prev}, the capacity has already been established when either parameter is zero. 
\end{rmk}
\medskip

We now define what it means for a code to achieve {\em secrecy} and {\em reliability} over the AWTC. We will require our codes to achieve both conditions. We will define both {\em weak} secrecy and the stronger notion of {\em semantic} secrecy. 
\smallskip

For a code $C\subseteq\{0,1\}^n$, we denote by $\mathbf{S}$ the random variable corresponding to the source message (distributed according to some distribution $P_{\mathbf{S}}$), and by $\mathbf{X}$ the random variable corresponding to $C$'s (possibly randomized) encoding of the message. We write $\hat{\mathbf{S}}$ for the output of the decoder upon receiving the corrupted version of $\mathbf{X}$. 

For a subset $\mathscr{S}\subseteq[n]$ of size $\rho_r n$ corresponding to the coordinates chosen by the adversary, denote by $\mathbf{V}(\mathscr{S})$ the view of the adversary after observing the coordinates in $\mathscr{S}$. That is, if $\mathscr{S}=\{i_1,\dotsc, i_{\rho_r n}\}\subseteq [n]$, and $x\sim \mathbf{X}$ is a codeword, then $\mathbf{V}(\mathscr{S})\in \{0,1,?\}^n$ is the string whose $i$th coordinate is $?$ if $i\notin \mathscr{S}$, and $x_i$ if $i\in \mathscr{S}$. We will also refer to $\mathscr{S}$ as the {\em support} of $\mathbf{V}(\mathscr{S})$. 

\begin{defn} Let $C\subseteq\{0,1\}^n$. The {\bf equivocation} $\Delta$ of the encoder is
\[\Delta(C) := \min_{\mathscr{S}:\abs{\mathscr{S}}=\rho_r n} H(\mathbf{S}| \mathbf{V}(\mathscr{S})). \]
\end{defn}

The quantity $H(\mathbf{S}| \mathbf{V}(\mathscr{S}))$ measures the uncertainty remaining after the coordinates of $\mathscr{S}$ have been observed, and higher $\Delta$ corresponds to higher secrecy (the adversary learns less about the encoded message).

\begin{defn}[Weak secrecy] \label{def:secrecy} 
Let $\mathcal{C}$ be a family of codes $(C_n\subseteq \{0,1\}^n)_{n\to\infty}$, and let the source distribution $P_{\mathbf{S}}$ be uniform. Let $\eta_n$ be the normalized equivocation of $C_n$; that is, 
\[\eta_n = \frac{1}{n} \Delta(C_n).\]

If $\eta_n$ approaches the rate $R(\mathcal{C})$ of $\mathcal{C}$ as $n\to\infty$, then we say that $\mathcal{C}$ achieves (asymptotic) {\bf weak secrecy}. 
\end{defn}

We now define semantic secrecy (security), via the equivalent notion of mutual-information security (see~\cite{semantic}). This notion of secrecy arises from dropping the assumption that messages are chosen uniformly from the message space. 

\begin{defn}[Semantic secrecy] \label{def:semantic} Let $\mathcal{C}$ be a family of codes $(C_n\subseteq \{0,1\}^n)_{n\to\infty}$. Define the semantic secrecy metric 
\[\mathrm{Sem}(C_n) = \max_{P_{\mathbf{S}}, \mathscr{S}} I(\mathbf{S}, \mathbf{V}(\mathscr{S})) = \max_{P_{\mathbf{S}},\mathscr{S}} D\bigl[P_{\mathbf{V}(\mathscr{S}) | \mathbf{S},C_n} \| P_{\mathbf{V}(\mathscr{S})| C_n} | P_{\mathbf{S}}\bigr], \]
where for two finite-support distributions $P$ and $Q$, $D(P\|Q)$ denotes the relative entropy
\[D(P\|Q) = \sum_{x\in \supp(P)} P(x) \log \frac{P(x)}{Q(x)}.\]

If $\mathrm{Sem}(C_n) = e^{-\Omega(n)}$, then we say that $\mathcal{C}$ achieves \bf{semantic secrecy}. 
\end{defn}

\begin{defn}[Reliability] Let $\mathcal{C}$ be a family of codes $(C_n\subseteq \{0,1\}^n)_{n\to\infty}$. Let $\delta_n$ be the average error probability of $C_n$; that is, 
\[\delta_n = \Pr[\hat{\mathbf{S}}\neq\mathbf{S}],\]
where the probability is taken over the randomness of the encoder and the error distribution introduced by the adversary. 

If $\delta_n\to 0$ as $n\to\infty$, then we say that $\mathcal{C}$ achieves (asymptotic) {\bf reliability}. 
\end{defn}

Let $\mathcal{C}$ be a code family of rate at least $R>0$ which achieves reliability. If $\mathcal{C}$ achieves weak secrecy, we say that $R$ is {\em achievable} under weak secrecy, and if $\mathcal{C}$ achieves semantic secrecy, we will say that $R$ is  achievable under semantic secrecy. 

\begin{defn} The {\bf weak secrecy capacity} of the $(\rho_r,\rho_w)$ AWTC is the supremum of achievable rates under weak secrecy, and the {\bf semantic secrecy capacity} is the supremum of achievable rates under semantic secrecy. 
\end{defn}

Our main result is the following. 
\begin{thm} \label{thm:main} Let $C_s(\rho_r,\rho_w)$ be the semantic secrecy capacity of the $(\rho_r,\rho_w)$ AWTC. Then
\[ \max\bigl(1 - h(\rho_w) - \rho_r,0\bigr)\leq C_s(\rho_r,\rho_w) \leq 1 - h(\rho_w) - \rho_r - \min_p f(p),\]
where $f(p) = h\bigl( (2\rho_w-1)p + 1 - \rho_w\bigr) - h(\rho_w) - \rho_r h(p).$

In fact, the upper bound holds under weak secrecy. 
\end{thm}

As there is a positive gap between our upper and lower bounds, some comparisons are provided in Figures~\ref{fig:bds} and~\ref{fig:ratio}. Although the bounds become far apart as $\rho_r$ approaches $1-h(\rho_w)$, because our lower bound approaches $0$ even as the capacity remains positive, when $\rho_r$ is small compared to $\rho_w$, we see that the two quantities are very close. 

\begin{figure}
\begin{center}
\includegraphics[width=4in]{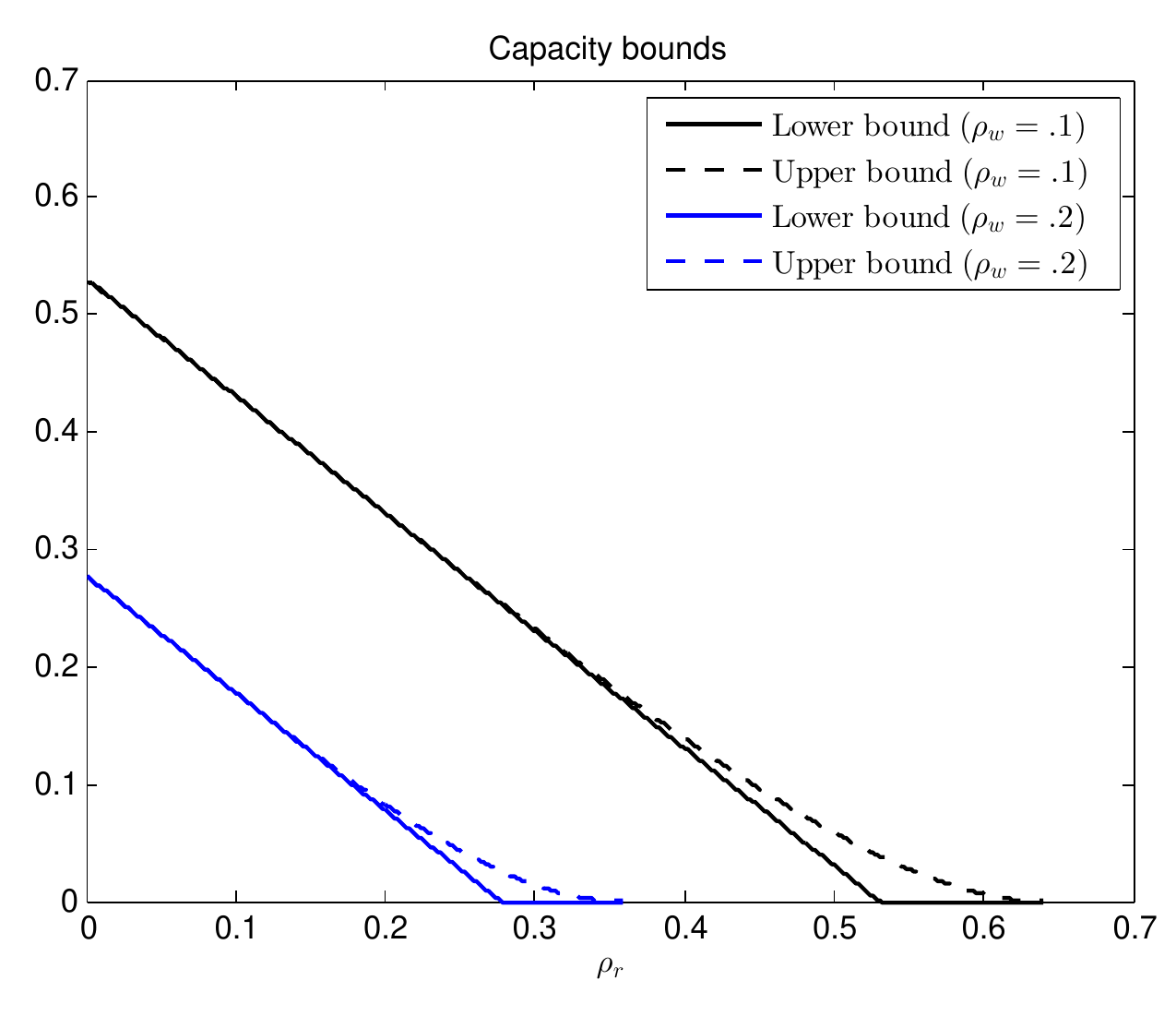}

\caption{\label{fig:bds}A comparison of our upper and lower bounds for fixed values of $\rho_w$. Note that the secrecy capacity is equal to zero when $\rho_r> 1 - 4\rho_w(1-\rho_w)$ (see~\cite{OU13}).}

\includegraphics[width=4in]{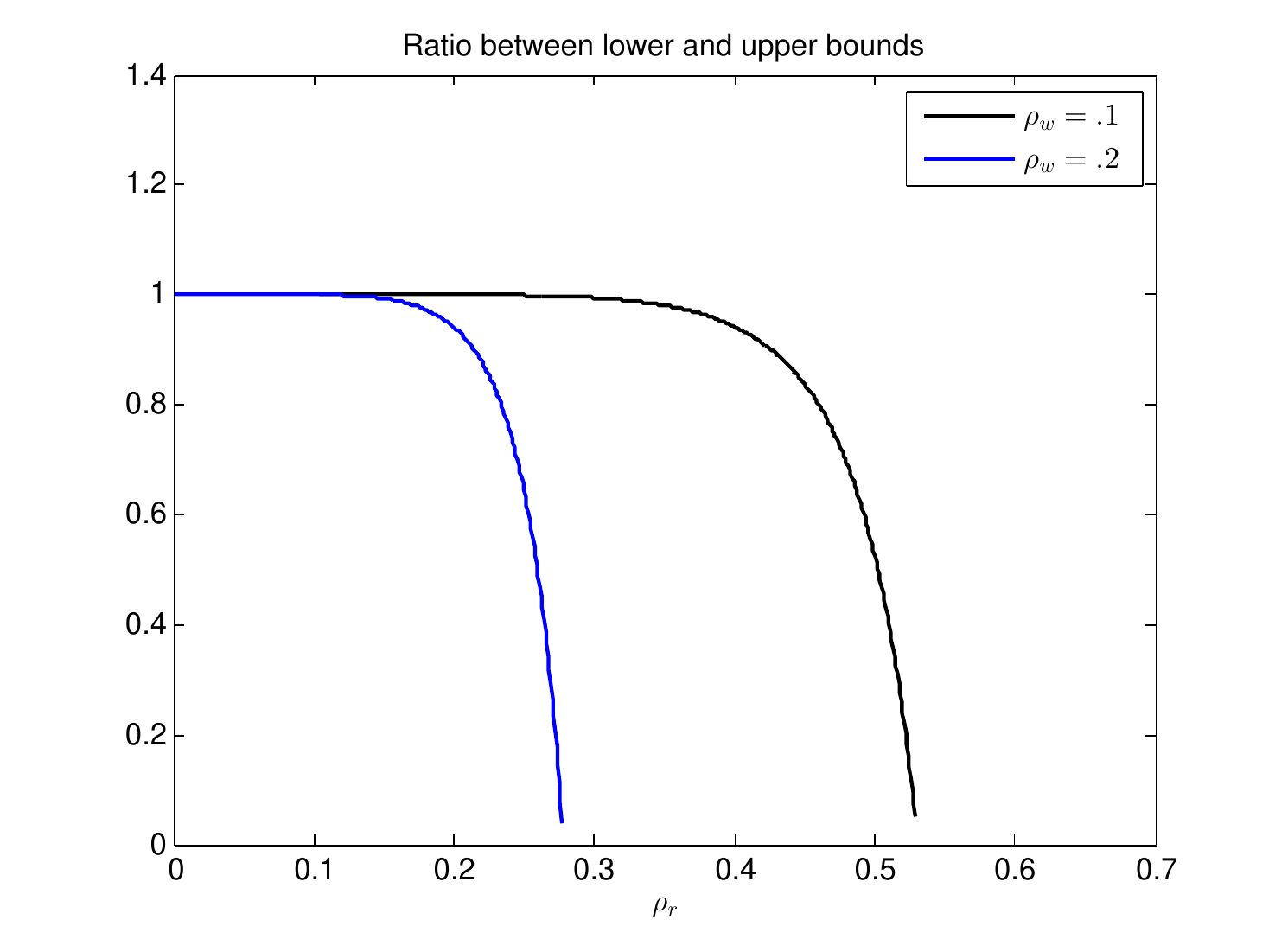}

\caption{\label{fig:ratio}The ratio between our lower and upper bounds, for fixed values of $\rho_w$. Note that our achievable lower bound is negative when $\rho_r > 1 - h(\rho_w)$. However, the two bounds are quite close when $\rho_r$ is small compared to $1-h(\rho_w)$. }
\end{center}
\end{figure}

\subsection{Previous work}
\label{sec:prev}

The (non-adversarial) wiretap channel has been the subject of a long line of work. Here we will focus only on variants which incorporate some sort of adversarial behavior. 
\medskip

\noindent {\bf Limited-view Adversaries and the AWTC.} The terminology ``adversarial wiretap channel'' for this model was introduced in~\cite{WSN}. The term ``limited-view adversary'' was also used for the same model in earlier work by the same authors. They give an upper bound on the capacity of the AWTP, and give an explicit construction of codes which meet this capacity over sufficiently large alphabets. 

\begin{rethm}[\cite{WSN}]  The semantic secrecy capacity of the $(\rho_r,\rho_w)$ adversarial wiretap channel is at most $1-\rho_r-\rho_w$. 

Moreover, when the alphabet size is $\abs{\Sigma}= \exp(\Omega(1/\epsilon^2))$, there is an explicit code of rate $1-\rho_r-\rho_w-\epsilon$ which achieves semantic secrecy and reliability over the $(\rho_r,\rho_w)$ AWTC. 
\end{rethm}

\begin{rmk} The construction of~\cite{WSN} not only achieves semantic secrecy, but in fact the stronger condition that for every message and choice of $\rho_r n$ coordinates, the distribution of the adversary's view is exactly uniform. This requires maximum distance separable (MDS) codes, which do not exist over binary alphabets. 
\end{rmk}

\medskip 
\noindent {\bf Wiretapping and Active Adversaries.} The authors of~\cite{wt2adv} consider the wiretap channel II model in which the adversary may also modify the bits which have been read. They then show that a suitable random code can be used to achieve (weak) secrecy and reliability. In fact, their construction also works without the restriction that the adversary modifies the same bits which he reads, so we may conclude the following. 

\begin{rethm}[\cite{wt2adv}] For any $\epsilon>0$ and $\rho_w<1/4$, there exist codes of rate 
\[1 - h(2\rho_w) - \rho_r-\epsilon\]
which achieve weak secrecy and reliability over the $(\rho_r,\rho_w)$ AWTC. 
\end{rethm}

As stated in the introduction, the difference between the $1-h(2\rho_w)-\rho_r$ and our achievable rate arises because the construction of~\cite{wt2adv} uses a binary code which corrects {\em any} $\rho_w$ fraction of errors, for which $1-h(2\rho_w)$ represents the best known achievable rate. The crux of our work is to show that errors introduced by a ``limited-view'' adversary behave more like random errors, for which $1-h(\rho_w)$ is the optimal rate. 
\medskip

\noindent {\bf The Wiretap Channel II.} 
When the adversary may read any $\rho_r$ fraction of the transmitted codeword, but does not inject any errors ($\rho_w=0$), this model is known as the {\em wiretap channel II}, due to Ozarow and Wyner~(\cite{wt2}). In this setting, we have the following result (a matching code construction is also presented). 

\begin{rethm}[\cite{wt2}] For any code of rate $R$ which achieves weak secrecy on the $(\rho_r,0)$ adversarial wiretap channel, we have
\[\rho_r\leq (1-R) + R h(\delta),\]
where $\delta$ is the error of the decoder. 

In particular, a code family of rate $R$ which achieves weak secrecy {\em and} reliability has $R\leq 1-\rho_r$. 
\end{rethm}

A generalization of this model is the Wiretap Channel II with a noisy main channel, studied in~\cite{NY}. Although they look at general classes of channels, one particular case which is relevant for us is when the main channel is a binary symmetric channel. In this case, we have
\begin{rethm}[implicit in~\cite{NY}] The weak secrecy capacity when the eavesdropper reads an arbitrary $\rho_r$ fraction of the transmitted codeword and the main channel is a $\mathrm{BSC}(\rho_w)$ is at most 
\[(1-\rho_r) (1-h(\rho_w)) = 1 - h(\rho_w) - \rho_r + \rho_r \cdot h(\rho_w).\]
\end{rethm}
\smallskip

The bounds of~\cite{NY} are tightened in~\cite{soft}, even under the semantic secrecy condition. The bound attained here is the same as that of the standard wiretap channel, stated in a modified form in Theorem~\ref{thm:rand-cap}. 

\begin{thm}[\cite{soft}] The semantic secrecy capacity when the eavesdropper reads an arbitrary $\rho_r$ fraction of the transmitted codeword and the main channel is a $\mathrm{BSC}(\rho_w)$ is equal to 
\[\max_{\substack{V-X-Y}}\bigl[I(V;Y) - \rho_r\cdot I(V;X)\bigr]^+,\]
where $X$ and $Y$ are the input and output, respectively, of the legitimate receiver, $[x]^+=\max(x,0)$, and the maximum is taken over all $(V,X,Y)$ such that $V-X-Y$ forms a Markov chain and the conditional distribution of $Y$ given $X$ is given by the $\mathrm{BSC}(\rho_w)$. 
\end{thm}

\medskip

\noindent {\bf Oblivious adversaries.} On the other hand, if $\rho_r=0$, the adversary is ``1-oblivious,'' and the capacity in this case has been established by Langberg in~\cite{obliv}. 

\begin{redefn}[\cite{obliv}] A binary channel $W$ is $\gamma$-{\bf oblivious} if $W$ imposes at most $2^{(1-\gamma)n}$ different error distributions $W(\cdot | \mathbf{x})$ over all $\mathbf{x}\in \{0,1\}^n$. 
\end{redefn}

In other words, an oblivious channel cannot use full information about the transmitted codeword $\mathbf{x}$, and only knows enough to determine {\em which} of the $2^{(1-\gamma) n} < 2^n$ error distributions to apply. For example, if the adversary reads an arbitrary $\rho_r n$ fraction of the codeword, then the number of possible error distributions is at most the number of such views, or at most $\binom{n}{\rho_r n}\cdot 2^{\rho_r n}<2^{(h(\rho_r) + \rho_r) n}$. 

\begin{rethm}[\cite{obliv}] For $\rho_w\in[0,1/2)$ and $\epsilon>0$, with high probability, a random code of rate 
\[R= \gamma - h(\rho_w) - \epsilon\] 
corrects a $\rho_w$ fraction of errors imposed by a $1$-oblivious channel with probability $1-2^{-\epsilon n+1}$ when codewords are chosen uniformly at random. 
\smallskip

It follows that the secrecy capacity of the $(0,\rho_w)$ AWTC is $1-h(\rho_w)$. 
\end{rethm}

The definition of a $\gamma$-oblivious channel is more general than that of the AWTC, as the distributions of the channel $W$ can depend on any aspect of the codeword $\mathbf{x}$, rather than just some fixed number of symbols. An adversary who reads only a $\rho_r$ fraction of symbols is $(1-\rho_r-h(\rho_r))$-oblivious; however, our result shows that we can obtain tighter rate bounds under our more stringent analogue of obliviousness. 

\medskip

\noindent{\bf Myopic Adversaries.} The work of~\cite{myopic} is closest in spirit to the current work. In their model, rather than reading a $\rho_r$ fraction of codeword symbols, the adversary receives the output of a $\mathrm{BSC}$ on the transmitted codeword. 

\begin{rethm}[\cite{myopic}, informal] When $\rho_r>\rho_w$, a random code of rate $1-h(\rho_w)-\epsilon$ can correct a $\rho_w$ fraction of errors with high probability when the error distribution depends only on $C$ and the output of a binary symmetric channel $\mathrm{BSC}(\rho_r)$ on the transmitted codeword.
\end{rethm}

This result is similar to what we want to prove, but the $\mathrm{BSC}$ gives less information to the adversary than a comparable erasure channel, and we are allowing the erasure channel to be {\em adversarial}. 

\section{Capacity upper bound}
\label{sec:ub}

To upper-bound the capacity of the adversarial wiretap channel, we reduce to the case of the standard, random wiretap channel~(\cite{wyner}). This channel, which we will refer to as the $(\rho_r,\rho_w)$ {\em random} wiretap channel, consists of a $\mathrm{BEC}(1-\rho_r)$ to the eavesdropper, and a $\mathrm{BSC}(\rho_w)$ to the receiver. The secrecy capacity of this channel is defined analogously to the previous section, and has been established in~\cite{OU13}. 

\begin{thm}[implicit in~\cite{OU13}] \label{thm:rand-cap}
When $\rho_r\leq 1- h(\rho_w)$, the secrecy capacity of the  $(\rho_r, \rho_w)$ random wiretap channel under the weak secrecy condition is equal to 
\[1 - h(\rho_w) - \rho_r -\min_p f(p),\]
where $f(p) = h\bigl( (2\rho_w-1)p + 1 - \rho_w\bigr) - h(\rho_w) - \rho_r h(p)$. 
\end{thm}

\begin{rmk} Readers familiar with the standard wiretap channel model may recognize the term $1 - h(\rho_w) - \rho_r$ (the difference of the capacities of the main and eavesdropper channels) as the correct capacity when the eavesdropper channel is noisier than the main channel in a formal sense. This condition holds when the channels are flipped; i.e.\ when the eavesdropper sees the output of a symmetric channel, and the main channel is a sufficiently good erasure channel, but for our case they are not comparable under this definition (see discussion in~\cite{OU13}). The work of~\cite{OU13} instead uses the notion of ``cyclic shift symmetric wiretap channels'' to obtain Theorem~\ref{thm:rand-cap}. Recall that $\min_p f(p)< 0$ for nonzero $\rho_w$ and $\rho_r$ (Figure~\ref{fig:bds}). 
\end{rmk}
\medskip

The reduction from the adversarial case to the random case is standard: the adversary can choose to read and write at random, subject only to the bound on the total number of errors. Thus any code for the adversarial wiretap channel must also be resilient to random errors. To the best of our knowledge, however, this reduction has not been given explicitly in the literature, so we provide a proof below. 

\begin{thm} \label{thm:rand-red}
Let $C_n\subseteq \{0,1\}^n$ be a code of rate $R$ which achieves decoding error $\delta_n$ and average equivocation $\eta_n$ over the $(\rho_r,\rho_w)$ adversarial wiretap channel. Then for sufficiently small $\xi>0$, $C_n$ achieves decoding error $\delta_n + \exp(-\Omega(n))$ and normalized equivocation 
$\eta_n(1 - \exp(-\Omega(n))$ over the $(\rho_r-\xi, \rho_w-\xi)$ random wiretap channel. 
\end{thm}
\begin{proof} 

{\bf Reliability}. By assumption, $C_n$ achieves decoding error $\delta_n$ over any distribution of errors with weight $\leq \rho_w n$. In particular, the decoding error of $C_n$ over the $\mathrm{BSC}(\rho_w-\xi)$ is at most 
\begin{align*}
    \delta_n + \Pr[\text{error has wt $>\rho_w$}] &= \delta_n + \sum_{i=\rho_wn+1}^n \binom{n}{i} (\rho_w-\xi)^{i} (1-\rho_w+\xi)^{n-i}\\
    &\leq \delta_n + 2^{-\Omega_{\rho_w,\xi}(n)}.
\end{align*}
\medskip

{\bf Weak secrecy}. We follow the notation of Section~\ref{sec:ntn}. Let $\mathbf{Z}\in \{0,1,?\}^n$ be the output of the eavesdropper's channel $\mathrm{BEC}(1-\rho_r + \xi)$ (with ``?'' denoting erasure). 

Let $\mathscr{S}$ be the set of coordinates which are not erased by the eavesdropper's channel (in earlier terminology, the support of $\mathbf{Z}$). For fixed $\mathscr{S}$, let $\mathbf{V}(\mathscr{S})$ be the output of the eavesdropper's channel conditioned on $\mathscr{S}$ being the support of $\mathbf{Z}$. Then the equivocation at the encoder is
\[\Delta(C_n) = H(\mathbf{S}|\mathbf{Z}) = \sum_{\mathscr{S}\subseteq[n]}\Pr[\supp(\mathbf{Z}) = \mathscr{S}]\cdot H(\mathbf{S}|\mathbf{V}(\mathscr{S})).\]

Let $E$ denote the event that $\abs{\mathscr{S}}\leq \rho_r n$. For $\xi>0$, we have $\Pr[E] \geq 1 - 2^{-\Omega_{\rho_r, \xi}(n)}$. By our assumption that $C_n$ achieves secrecy over the AWTC, whenever $\abs{\mathscr{S}}=\rho_r n$, the equivocation is $H(\mathbf{S} | \mathbf{V}(\mathscr{S}))\geq \eta_nn$, where $\eta_n\to R$. When $\mathscr{S}$ has size less than $\rho_r n$, we have $H(\mathbf{S} | \mathbf{V}(\mathscr{S}))\geq H(\mathbf{S} | \mathbf{V}(\mathscr{S}'))$, for arbitrary $\mathscr{S}'\supseteq \mathscr{S}$ of size $\rho_r n$ (in other words, conditioning does not increase entropy). In particular, for any $\mathscr{S}$ of size $\leq \rho_r n$, $H(\mathbf{S} | \mathbf{V}(\mathscr{S}))\geq \eta_nn$. 

Thus, the normalized equivocation of the decoder is at least 
\begin{align*}
\frac{\Delta}{n} &\geq \Pr[E]\cdot \eta_n + \frac{1}{n}\Pr[\neg E]\cdot H(\mathbf{S} | \neg E,\mathbf{Z})\\
    &\geq \eta_n\cdot \bigl(1-\exp(-\Omega_{\rho_r, \xi}(n))\bigr).
\end{align*}

As $\eta_n\to R$ as $n\to\infty$, we have $\Delta/n = R - o(1)$. 
\smallskip

Thus, we have shown that $C$ achieves reliability and secrecy over the $(\rho_r-\xi,\rho_w-\xi)$ random wiretap channel, as desired. 

\end{proof}

\begin{cor} \label{cor:upper}
The weak secrecy capacity of the $(\rho_r,\rho_w)$ adversarial wiretap channel is at most $1 - h(\rho_w) - \rho_r -\min_p f(p),$
where $f(p)$ is defined as in Theorem~\ref{thm:rand-cap}. 
\end{cor}
\begin{proof} Let $\mathcal{C}$ be a family of codes achieving a rate $R$ over the $(\rho_r,\rho_w)$ AWTC. By Theorem~\ref{thm:rand-red}, for sufficiently small $\xi>0$, $C$ also achieves the rate $R$ over the $(\rho_r-\xi, \rho_w-\xi)$ random wiretap channel. Thus, by Theorem~\ref{thm:rand-cap}, $R\leq 1 - h(\rho_w-\xi) - \rho_r-\min_p f(p)+\xi$. 

As $\xi$ can be arbitrarily small, we must have $R\leq 1 - h(\rho_w) - \rho_r - \min_p f(p)$. 
\end{proof}

\section{Capacity lower bound}
\label{sec:lb}

We will show that the semantic secrecy capacity of the $(\rho_r,\rho_w)$ AWTC is at least $1-h(\rho_w) - \rho_r$ by giving a code construction which achieves this rate. Unsurprisingly, we will show that a random stochastic code works. (Recall that a stochastic code uses a probabilistic encoding function, allowing multiple codewords to correspond to the same message.) 
\medskip

For $\epsilon>0$, let $C\subseteq\{0,1\}^n$ be a random code of rate $R:=1-h(\rho_w) - \epsilon$ constructed by selecting $2^{Rn}$ i.i.d.~vectors uniformly from $\{0,1\}^n$. Abusing notation, we also refer to the encoding map $C\colon \{0,1\}^{Rn}\to \{0,1\}^n$ which maps a binary string of length $Rn$ to its corresponding codeword. 

Define $R'$ such that $R'n = Rn - \ell$, for some $\ell$ to be set later. In order to ensure {\em secrecy} for our construction, we define the following partition $\Pi=\{A_i\}_{i\in\{0,1\}^{R'n}}$ of $C$, which splits $C$ into sets of size $2^\ell$. 

We may then consider $C$ as a stochastic code
\[C_\Pi\colon \{0,1\}^{R'n} \times \{0,1\}^{\ell}\to\{0,1\}^n\]
of rate $R' = 1 - h(\rho_w) - \ell/n - \epsilon$, which maps the $i$th message $i\in \{0,1\}^{R'n}$ and a random index (seed) $r\in\{0,1\}^\ell$ to $C(m,r)$, the $r$th codeword in $A_i$. 

\begin{rmk} We will show below that we can choose $\ell= \rho_r n - \Theta(1)$ in order to achieve weak secrecy, so that the final rate of $C_\Pi$ can be taken to be $1 - h(\rho_w) - \rho_r - \epsilon$. 

For semantic secrecy, we will show that any $\ell/n > \rho_r$ suffices, so that the final rate of $C_\Pi$ can be taken to be, for example, $1 - h(\rho_w) - \rho_r - \epsilon/2$. 

This rate bound, combined with the results of Corollary~\ref{cor:secrecy} and Theorem~\ref{thm:reliability}, which prove secrecy and reliability for this code, show that the capacity of the $(\rho_r,\rho_w)$ AWTC is at least $1 - h(\rho_w) - \rho_r$, as desired. 
\end{rmk}
\medskip

In the following sections, we will show that $C_\Pi$ achieves both secrecy and reliability over the $(\rho_r,\rho_w)$ adversarial wiretap channel (precise statements below).

\subsection{Secrecy}
\label{sec:secrecy-proof}

In this section, we give two proofs that our random code construction achieves secrecy. The first, based on a simple counting argument, shows that weak secrecy is achieved. The second, based on a more sophisticated probabilistic argument from~\cite{soft}, shows that semantic secrecy is achieved. Although the first result is weaker, we believe it is more transparent. Furthermore, although the asymptotic rates achieved for both are the same, the weak secrecy proof allows for a slightly higher rate at fixed block lengths than we are able to show in the semantic secrecy case. 

\subsubsection{Weak secrecy via counting}

We show that the random code $C_\Pi$ achieves weak secrecy with high probability. Following the approach of~\cite{wt2} (which is also used in~\cite{wt2adv}), we bound the equivocation of the encoder. Recall that the rate-$R'$ code $C_\Pi\subseteq\{0,1\}^n$ encodes a message $m_i$ by a random vector $x\in A_i$, where $\Pi=\{A_i\}$ is some partition of $C$. 

As before, we denote by $\mathbf{S}$ the random variable corresponding to the (uniformly chosen) source message, and $\mathbf{X}$ the random variable corresponding to $C_\Pi$'s (randomized) encoding of the message. For a subset $\mathscr{S}\subseteq[n]$ of size $\rho_r n$ corresponding to the coordinates chosen by the adversary, denote by $\mathbf{V}(\mathscr{S})$ the view of the adversary on $\mathscr{S}$. 

Recall that the equivocation of the encoder is $\Delta := \min_{\mathscr{S}:\abs{\mathscr{S}}=\rho_r n} H(\mathbf{S}| \mathbf{V}(\mathscr{S}))$. In the best case, $\Delta = R'n$, the dimension of the code, and nothing is learned. 

We will show that with high probability, the normalized equivocation $\Delta/n$ approaches $R'$ as $n\to\infty$. 
\medskip

In what follows, we write $\mathbf{V}$ for $\mathbf{V}(\mathscr{S})$, where $\mathscr{S}$ minimizes the equivocation of the encoder. 

\begin{lemma}[\cite{wt2}] The equivocation $\Delta$ of the encoder corresponding to $C_\Pi$ is 
\begin{align*}
\Delta & = H(\mathbf{S} | \mathbf{X},\mathbf{V}) + H(\mathbf{X}|\mathbf{V}) - H(\mathbf{X} | \mathbf{S}, \mathbf{V})\\
&=H(\mathbf{X}|\mathbf{V}) - H(\mathbf{X} | \mathbf{S}, \mathbf{V})\\
&\geq  Rn - \rho_r n - H(\mathbf{X}|\mathbf{S}, \mathbf{V}).
\end{align*}
\end{lemma}
\medskip

Given the above, we are interested in ensuring that $H(\mathbf{X}|\mathbf{S}, \mathbf{V})$ is small. In other words, given the message $m_i$, there should be few codewords in the codeword set $A_i$ which are consistent with any fixed view. 

\begin{defn} Let $(v_1,\dotsc, v_n)\in \{0,1,?\}^n$ be the view of the adversary. A codeword $x\in C$ is {\bf consistent} with $V$ if $x_i=v_i$ whenever $v_i\neq ?$. 
\end{defn} 

\begin{cor} \label{cor:equiv} Let $L\geq 1$. If for every view $V$ and every message index $i$, we have
\begin{equation}
\label{eq:cons}
\#\{x\in A_i\mid x~\text{ is consistent with }V\}<L,
\end{equation}
then 
\[\Delta \geq Rn - \rho_r n - \log L.\]
\end{cor}

We will show sufficient conditions on $L$ and $\ell$ for the condition of Equation~\eqref{eq:cons} to hold, and use these to show that a normalized equivocation $\Delta/n$ going to $R'$, the rate of the code, is achievable. More specifically, 

\begin{lemma}\label{lem:secrecy}
Assume $R'>0$. If $\ell/n < \rho _r - 2(R'+1)/L$, then with high probability over the choice of the code $C$, 
\[\#\{x\in A_i\mid x~\text{is consistent with }V\}<L\]
for all $V$, $i$. 
\end{lemma}
\begin{proof} Fix $V$ and the index $i$. The codeword set $A_i$ consists of $2^\ell$ i.i.d.~vectors chosen uniformly from $\{0,1\}^n$. 

Let $S|_V\subseteq\{0,1\}^n$ be the set of strings which is consistent with $V$. Then $S|_V$ has size $2^{(1-\rho_r) n}$. 
The number of subsets of $S_V$ of size $L$ is 
\[\binom{2^{n-\rho_r n}}{L}\leq 2^{n(1-\rho_r) L}.\]

The probability that any fixed subset of size $L$ is contained in $A_i$ is at most $2^{-(n-\ell) L}$. 

Thus the probability that 
\[\#\{x\in A_i\mid x~\text{is consistent with $V$}\}\geq L\]
is at most 
\[2^{n(1-\rho_r) L}\cdot 2^{-(n-\ell) L}.\]
Union bounding over the $2^{nR'}$ choices of $i$ and the $2^{\rho_r n}\cdot \binom{n}{\ell}\leq 2^{\rho_r n}2^{n h(\ell/n)}$ choices of $V$, the probability that the code $C_\Pi$ fails (i.e.\ that $\#\{x\in A_i\mid x~\text{is consistent with $V$}\}\geq L$ for some $i$ and some $V$) is 
\[\Pr[C_\Pi~\text{ fails}]\leq 2^{n(1-\rho_r) L - (n-\ell) L + \rho_r n + nR' + n h(\ell/n)}.\]

Rearranging, we have
\[\frac{\log(\Pr[C_\Pi~\text{ fails}])}{n} \leq - \rho_r(L-1) + \frac{\ell L}{n} + R' + h(\ell/n).\]
If we set
\[\frac{\ell}{n} < \rho_r - \frac{2(R' + 1)}{L},\]
then the probability that $C_\Pi$ fails is at most $2^{-R' n}=2^{-\Omega(n)}$, as desired. 
\end{proof}

\begin{cor} \label{cor:secrecy}
With high probability over the choice of the code $C$, the normalized equivocation of the encoder approaches $R'$, the rate of the code $C_\Pi$, as the block length $n$ goes to infinity. 
\end{cor}
\begin{proof} In Lemma~\ref{lem:secrecy}, set $L=n$, and set $\ell = \rho_r n -2(R'+1)$. 

Then by Lemma~\ref{lem:secrecy}, $C$ satisfies the condition of Corollary~\ref{cor:equiv} with $L=n$. Thus, the equivocation of the code is
\begin{align*} 
\Delta&\geq Rn - \rho_rn - \log n -O(1)\\
    & = R'n +\ell - \rho_r n - \log n-O(1)\\
    &\geq R' n - 2(R'+1) - \log n - O(1).
\end{align*}

Thus, the normalized equivocation is 
\[\frac{\Delta}{n} \geq R' - O(\log n/n),\]
which approaches $R'$ as $n\to\infty$, as desired. 
\end{proof}

\subsubsection{Semantic secrecy via soft-covering}

Semantic secrecy for the Wiretap Channel II with noisy main channel is established in~\cite{soft}. Although their reliability proof does not apply to our setting with an adversarial main channel, the adversarial nature of the Wiretap Channel II means that the secrecy analysis applies to the AWTC. We outline the proof below; for more details, see the original derivation in~\cite{soft}. 
\medskip

At the core of the analysis is the following ``stronger soft-covering lemma,'' which shows that the distribution of the output of the channel on a large random subset of codewords is unlikely to be far from the distribution of the output of the channel on a truly random codeword. In other words, if each message is associated to $\approx 2^{(\rho_r + \epsilon) n}$ codewords, then the adversary's view will be statistically close to uniform for every message, and secrecy is achieved. 
\smallskip

In the following, $Q_U$ is the uniform distribution on $\{0,1\}$; $Q_{V|U}$ is a memoryless channel with output alphabet $\mathcal{V}$, and $S_n$ is a set of $2^{\hat{R}n}$ independent elements of $\{0,1\}^n$ chosen uniformly at random. For fixed $S_n$, $P_{\mathbf{V}|S_n}$ denotes the distribution on $\mathcal{V}^n$ induced by $Q_{V|U}^n$ on a uniform element of $S_n$; that is, 
\[P(\mathbf{v} | S_n) = 2^{-\hat{R} n} \sum_{w\in S_n} Q^n_{V|U} (\mathbf{v} | w).\]

\begin{lemma}[\cite{soft}] \label{lem:soft}
For any $Q_{V|U}$, and $\hat{R}>I(U;V)$, where $\abs{\mathcal{V}}<\infty$, there exist $\gamma_1,\gamma_2>0$ such that for large enough $n$, 
\[\Pr\bigl[D(P_{\mathbf{V}|S_n} \| Q_V^n)> e^{-n\gamma_1}\bigr] \leq e^{-e^{n\gamma_2}},\]
where $D(\cdot\|\cdot)$ is the relative entropy. 
\end{lemma}

In what follows, we will apply this lemma to the set of codewords associated to any fixed message, showing that the adversary's view is always nearly indistinguishable from uniform. 
\medskip

Recall that our goal is to bound the semantic secrecy metric
\[\mathrm{Sem}(C_n) = \max_{P_{\mathbf{S}},\mathscr{S}} D\bigl[P_{\mathbf{V}(\mathscr{S}) | \mathbf{S},C_n} \| P_{\mathbf{V}(\mathscr{S})| C_n} | P_{\mathbf{S}}\bigr]\]
where $P_{\mathbf{S}}$ is any distribution over the source message $\mathbf{S}$, and $\mathscr{S}$ is a subset of $[n]$ of size $\rho_r n$. 

First we fix $\mathscr{S}$ and $P_{\mathbf{S}}$, in which case (\cite{soft})
\[D\bigl[P_{\mathbf{V}(\mathscr{S}) | \mathbf{S},C_n} \| P_{\mathbf{V}(\mathscr{S})| C_n} | P_{\mathbf{S}}\bigr]\leq D\bigl[P_{\mathbf{V}(\mathscr{S}) | \mathbf{S},C_n} \| P_{\mathbf{Z}(\mathscr{S})}|P_{\mathbf{S}}\bigr],\]
where $P_{\mathbf{Z}(\mathscr{S})}$ is the uniform distribution over all strings in $\{0,1,?\}^n$ with support equal to $\mathscr{S}$. 
\smallskip

Maximizing over $P_{\mathbf{S}}$ and recalling that the source message is taken from $\{0,1\}^{R'n}$, we see that 
\[\max_{P_{\mathbf{S}}}D\bigl[P_{\mathbf{V}(\mathscr{S}) | \mathbf{S},C_n} \| P_{\mathbf{Z}(\mathscr{S})}|P_{\mathbf{S}}\bigr] \leq \max_{s\in \{0,1\}^{R'n}} D\bigl[P_{\mathbf{V}(\mathscr{S}) | \mathbf{S}=s,C_n} \| P_{\mathbf{Z}(\mathscr{S})}\bigr],\]
so in particular
\begin{equation}\label{eq:semsec}
\max_{P_{\mathbf{S}},\mathscr{S}} D\bigl[P_{\mathbf{V}(\mathscr{S}) | \mathbf{S},C_n} \| P_{\mathbf{V}(\mathscr{S})| C_n} | P_{\mathbf{S}}\bigr]\leq \max_{s\in S, \mathscr{S}} D\bigl[P_{\mathbf{V}(\mathscr{S}) | \mathbf{S}=s,C_n} \| P_{\mathbf{Z}(\mathscr{S})}\bigr]
\end{equation}
\medskip

Let us fix $\delta>0$, to be set later. We would like to bound the probability that $C_n$ is not semantically secure, which we do by considering the probability of the following event:
\[\mathrm{Sem}(C_n) = \max_{P_{\mathbf{S}},\mathscr{S}} D\bigl[P_{\mathbf{V}(\mathscr{S}) | \mathbf{S},C_n} \| P_{\mathbf{V}(\mathscr{S})| C_n} | P_{\mathbf{S}}\bigr]> e^{-n\delta}.\]

By Equation~\eqref{eq:semsec}, 
\begin{equation}\label{eq:sembds}
\Pr\bigl[\mathrm{Sem}(C_n)>e^{-n\delta}\bigr]\leq \Pr\left[\max_{s\in S, \mathscr{S}} D\bigl[P_{\mathbf{V}(\mathscr{S}) | \mathbf{S}=s,C_n} \| P_{\mathbf{Z}(\mathscr{S})}\bigr] > e^{-n\delta}\right].
\end{equation}

In order to apply the stronger soft-covering lemma, we fix $s$ and $\mathscr{S}$. Note that $\mathbf{V}(\mathscr{S})$ and $\mathbf{Z}(\mathscr{S}$ are supported only on strings with ``?'' in coordinates not in $\mathscr{S}$. Denoting by $P^{\mathscr{S}}_{\mathbf{V}(\mathscr{S}) | \mathbf{S}=s,C_n}$ and  $P^{\mathscr{S}}_{\mathbf{Z}(\mathscr{S})}$ the respective distributions restricted to coordinates in $\mathscr{S}$, the relative entropy chain rule implies that 
\begin{equation}\label{eq:rest}
D\bigl[P_{\mathbf{V}(\mathscr{S}) | \mathbf{S}=s,C_n} \| P_{\mathbf{Z}(\mathscr{S})}\bigr] = D\bigl[P^{\mathscr{S}}_{\mathbf{V}(\mathscr{S}) | \mathbf{S}=s,C_n} \| P^{\mathscr{S}}_{\mathbf{Z}(\mathscr{S})}\bigr].
\end{equation}

(In other words, the ``?'' coordinates outside of $\mathscr{S}$ do not affect the relative entropy on the left-hand side of Equation~\eqref{eq:rest}.)

Recall that $P^{\mathscr{S}}_{\mathbf{Z}(\mathscr{S})}$ is just a uniform distribution over strings in $\{0,1\}^{\abs{\mathscr{S}}}$. In the statement of Lemma~\ref{lem:soft}, we may set $Q_{V|U}$ to be the binary identity channel (so $\mathcal{V}=\{0,1\}$). 

$P^{\mathscr{S}}_{\mathbf{V}(\mathscr{S}) | \mathbf{S}=s,C_n}$ is a distribution on $\{0,1\}^{\abs{\mathscr{S}}}$ induced by applying $Q_{V|U}$ to the $2^{\ell}$ uniformly random codewords which encode $s$ (restricted to the coordinates in $\mathscr{S}$). Thus, if $\ell/(\rho_r n)>H(U) = 1$, then by Lemma~\ref{lem:soft}, there exist $\gamma_1,\gamma_2>0$ such that for large enough $n$, 

\[\Pr\biggl[D\bigl[P^{\mathscr{S}}_{\mathbf{V}(\mathscr{S}) | \mathbf{S}=s,C_n} \| P^{\mathscr{S}}_{\mathbf{Z}(\mathscr{S})}\bigr]> e^{-\gamma_1}\biggl]\leq e^{-e^{n\gamma_2}}.\]
\medskip

Setting $\delta = \gamma_1$ in Equation~\eqref{eq:sembds}, we have that 
\begin{align*}
\Pr\bigl[\mathrm{Sem}(C_n)>e^{-n\delta}\bigr]&\leq \Pr\left[\max_{s\in S, \mathscr{S}} D\bigl[P_{\mathbf{V}(\mathscr{S}) | \mathbf{S}=s,C_n} \| P_{\mathbf{Z}(\mathscr{S})}\bigr] > e^{-n\delta}\right]\\
    &\leq \sum_{s,\mathscr{S}}\Pr\biggl[D\bigl[P^{\mathscr{S}}_{\mathbf{V}(\mathscr{S}) | \mathbf{S}=s,C_n} \| P^{\mathscr{S}}_{\mathbf{Z}(\mathscr{S})}\bigr]> e^{-\gamma_1}\biggl]\tag{union bound}\\
    &\leq 2^{n}\cdot 2^n\cdot e^{-e^{\Omega(n)}}\\
    &=e^{-\Omega(n)}.
\end{align*}

Thus, $C_n$ achieves semantic secrecy with high probability when $\ell/n>\rho_r$. 

\subsection{Reliability}
\label{sec:reliability-proof}

In this section, we show that a random code $C$ of rate $1-h(\rho_w)-\epsilon$ can correct a $\rho_w$ fraction of errors with high probability (over the choice of the code and the uniform choice of the transmitted codeword) when the error distribution depends only on $C$ and a $\rho_r$ fraction of the transmitted codeword. This result is similar in spirit to those of~\cite{obliv} and~\cite{myopic}, described in Section~\ref{sec:prev}. 

More formally, we will show the following.  
\begin{thm} \label{thm:reliability} Fix $\rho_r,\rho_w>0$. Let $C\subseteq\{0,1\}^n$ consist of $2^{Rn}$ i.i.d.~uniform vectors, where $R=1 - h(\rho_w) - \epsilon$ for some $\epsilon>0$. Then if $\rho_r  < 1 - h(\rho_w)-2\epsilon$, $C$ achieves decoding error $o(1)$ over the adversarial $(\rho_r,\rho_w)$ wiretap channel with high probability. 
\end{thm}

We follow the proof outline of~\cite{myopic} in order to show that the random code $C$ achieves reliability over the adversarial $(\rho_r,\rho_w)$ wiretap channel. The primary difference is in the definition of the set $C|_V$, as our adversarial model is different. 

In this section, we will not use the association of codewords to messages induced by the partitioning $\Pi$, and will only consider $C$ as an (ordered) codebook of $2^{Rn}$ uniform, equally likely strings. 

Once the adversary has read a $\rho_r n$ fraction of the transmitted codeword, he knows that this codeword lies in the subset of $C$ which is {\em consistent} with the read symbols, but each element of that subset is equally likely from his perspective. We will show that this is sufficient to ensure that he cannot cause a decoding error with non-vanishing probability. 
\medskip

The main tool in our analysis is the following lemma, from~\cite{myopic}. At a high level, it states that a set of random vectors (in our case, codewords) cannot be too concentrated in a small volume $B$. 
\begin{lemma} \label{lem:conc}
Let $A\subseteq\{0,1\}^n$ be a set with $2^{\alpha n}$ elements for some $\alpha>0$, $\nu, \beta>0$, and $B\subset A$ with $\abs{B}\leq 2^{n(\alpha-\beta-\nu)}$. Let $X_1,\dotsc, X_N$ be chosen uniformly at random from $A$ with $N=2^{n\beta}$. Then for constant $c>0$, 
\[\Pr[\#\{ i\mid X_i\in B\}>cn^2]\leq \exp\bigl(-\Omega(n^2)\bigr).\]
\end{lemma}

This lemma, which follows from a Chernoff bound, implies a standard result on the list-decodability of random codes.

\begin{cor} \label{cor:ld} Let $C\subseteq\{0,1\}^n$ be a random code of rate $1-h(\rho_w) - \epsilon$ for $\epsilon>0$. 
With probability $1-\exp(-\Omega(n^2))$ over the choice of $C$, every Hamming ball in $\{0,1\}^n$ of radius $\rho_w n$ contains at most $O(n^2)$ elements of $C$. 
\end{cor}
\begin{proof} In Lemma~\ref{lem:conc}, let $B\subseteq\{0,1\}^n$ be a Hamming ball of radius $\rho_w n$, which has size at most $2^{nh(\rho_w)}$. A union bound over all $2^n$ such Hamming balls proves the statement. 
\end{proof}
\medskip

As before, let us denote the adversary's view by $(v_1,\dotsc, v_n)\in \{0,1,?\}^n$. 

Denote by $C|_V$ the set of codewords which are consistent with $V$ (i.e., the codewords which have $v_i$ in the $i$th coordinate, whenever $v_i\neq ?$). Let $E_0$ be the event that $C|_V$ has size between $2^{(1-h(\rho_w) - \rho_r - 3\epsilon/2)n}$ and $2^{(1-h(\rho_w) - \rho_r - \epsilon/2)n}$ for all views $V$. Using a Chernoff bound and union-bounding over all $V$, we see that 

\[\Pr_C[E_0]\geq 1-\exp(-2^{\Omega(n)}).\]

Conditioned on $E_0$, $C|_V$ has size at least $2^{\epsilon n/4}$, by our assumption on $\rho_r$. We partition $C|_V$ into sets $S_i$ of size $2^{\epsilon n/4}$, with the exception of the last set, which may have size less than $2^{\epsilon n/4}$. This is done using the ordering on $C$: the first set $S_1$ consists of the first $2^{\epsilon n/4}$ elements of $C|_V$, and each subsequent block contains the next $\leq 2^{\epsilon n/4}$ elements of $C|_V$. 

Let $E_1$ be the event that the transmitted codeword does not lie in the last set of the partition. Over the uniform choice of the transmitted codeword, we have that 
\[\Pr[E_1|E_0]\geq 1-2^{\epsilon n/4}/\abs{C|_V}\geq 1 - \exp\bigl(-\Omega(n)\bigr).\]

Following the approach of~\cite{myopic}, we will show that (conditioned on $E_0$ and $E_1$) for any fixed error vector $e$ of weight $\rho_w n$, for every $i$, the probability over $x\in S_i$ that $x+e$ causes a decoding error is small. In other words, even if the adversary is only oblivious over the much smaller set $S_i\subseteq C|_V$, he still cannot reliably cause a decoding error. (We can think of the set $S_i$ which contains the true transmitted codeword as being side or oracle information given to the decoder.) 

In the analysis which follows, we may think of the set $S_r$ as one of the sets $S_i$ which have size $2^{\epsilon n/4}$. We will show that any such set $S_r$ causes few decoding errors. 

\begin{defn} Let $x\in \{0,1\}^n$, and let $e\in \{0,1\}^n$ be fixed of weight at most $\rho_w n$. Denote by $B_{\rho_w} ( x + e)$ the Hamming ball of radius $\rho_w n$ around $x+e$. 

If $x'\in B_{\rho_w}(x+e)$, we say that $x+e$ {\bf conflicts} with $x'$. 
\end{defn}

In other words, if $x+e$ conflicts with $x'$, then both $x$ and $x'$ are valid outputs of the decoder, and we will consider this to be a decoding error. In what follows, we will show that for any fixed error vector $e$, there are ``few'' conflicts in a random code. 

\begin{lemma} \label{lem:external}
Let $S_r\subseteq C|_V$ be of size $2^{\epsilon n/4}$, and let $e\in\{0,1\}^n$ have weight at most $\rho_w n$. Let $N_e$ denote the number of codewords $x\in S_r$ such that $x+e$ conflicts with some codeword {\em not} in $S_r$. Then, conditioned on $E_0$, $N_e\leq O(n^4)$ with probability $1-\exp(-\Omega(n^2))$ over the choice of $C$. 
\end{lemma}
\begin{proof} We first show that $S_r$ has few conflicts with the codewords {\em outside of} $C|_V$. Let $A\subseteq\{0,1\}^n$ be the set of all strings with are {\em not} consistent with $V$. In Lemma~\ref{lem:conc}, set $B=\bigcup_{x\in S_r} B_{\rho_w} ( x + e)$, and let $X_1,\dotsc, X_N$ be the codewords of $C\setminus C|_V$. By construction of $C$, the $X_i$ are chosen uniformly from the space $A$, and we have $N\leq 2^{nR}$. 

We have $\abs{A}=2^n-2^{(1-\rho_r) n}$, and $\abs{B} \leq 2^{h(\rho_w) n + \epsilon n/4}\leq 2^{(1-R-\epsilon/2)n}$. Thus, by Lemma~\ref{lem:conc}, with probability $1-2^{-\Omega(n^2)}$, there are at most $O(n^2)$ codewords $x'\in C\setminus C|_V$ such that $x'$ is at distance $\leq \rho_w n$ from $x+e$, for some $x\in S_r$. 

By Corollary~\ref{cor:ld}, with probability $1-\exp\bigl(-\Omega(n^2)\bigr)$, for each such $x'$, the number of codewords in $C$ in the Hamming ball of radius $\rho_w n$ around $x'-e$ is at most $O(n^2)$. 
\medskip

Now we show that $S_r$ has few conflicts with the codewords in $C|_V\setminus S_r$. Because we are conditioning on $E_0$, we have that $\abs{C|_V}\leq 2^{(1-h(\rho_w) - \rho_r - \epsilon/2)n}$. Set $B$ as before, and let $A^c\subseteq\{0,1\}^n$ be the set of strings which are consistent with $V$ (the complement of the set $A$ above). We have $\abs{A^c} = 2^{(1-\rho_r) n}$. 

As the elements of $C|_V\setminus S_r$ are uniformly distributed in the set $A^c$, we may apply Lemma~\ref{lem:conc} once more. As $\abs{B\cap A^c}\leq 2^{h(\rho_w) n + \epsilon n/4}\leq 2^{(1-\rho_r) n - (1-h(\rho_w) - \rho_r - \epsilon/2)n - \epsilon n/4}$, by Lemma~\ref{lem:conc} we conclude that with probability $1-\exp\bigl(-\Omega(n^2)\bigr)$, there are at most $O(n^2)$ codewords $x'\in C|_V\setminus S_r$ such that $x'$ is at distance $\rho_w n$ from $x+e$, for some $x\in S_r$. 

Again by Corollary~\ref{cor:ld}, with probability $1-\exp\bigl(-\Omega(n^2)\bigr)$, for each such $x'$, the number of codewords in $C$ in the Hamming ball of radius $\rho_w n$ around $x'-e$ is at most $O(n^2)$. 
\medskip

This shows that there are at most $O(n^2)$ codewords $x'\in C\setminus S_r$ such that $x'\in \bigcup_{x\in S_r} B_{\rho_w} ( x + e)$, and each such $x'$ is in at most $O(n^2)$ balls $B_{\rho_w} ( x + e)$. Thus conditioned on $E_0$, $N_e\leq O(n^4)$ with probability $1-\exp(-\Omega(n^2))$, as desired. 
\end{proof}

It remains to show that when the transmitted codeword is chosen uniformly from $S_r$, the adversary is unlikely to cause confusion with another codeword in $S_r$. This can be shown directly using the techniques of~\cite{myopic}, or by appealing to the results of~\cite{obliv}. Let us outline the argument of~\cite{myopic}. 

\begin{lemma}\label{lem:subset}
Assume that $E_0$ holds. Let $T\subseteq \{0,1\}^n$ have size $2^{\epsilon n/8}$, and fix $e\in\{0,1\}^n$ of weight at most $\rho_w n$. Let $S_r\subseteq C|_V$ have size $2^{\epsilon n/4}$. 

Then with probability $1-\exp(-\Omega(n^2))$ over the choice of $S_r$, the number $M_e$ of $x\in T$ such that $x+e$ conflicts with some codeword in $S_r$ is at most $O(n^4)$. 
\end{lemma}
\begin{proof} Note that the elements of $S_r$ are chosen uniformly at random from a space of size $2^{(1-\rho_r) n}$, and the size of $S_r$ is $2^{\epsilon n/4}$. 

Similarly to the previous lemma, we apply Lemma~\ref{lem:conc} with $B=\bigcup_{x\in T} B_{\rho_w n}(x+e)$. We have $\abs{B}\leq 2^{h(\rho_w) n + \epsilon n/8} < 2^{(1-\rho_r) n - \epsilon n/4}$ (recall that $\rho_r < 1-h(\rho_w)-2\epsilon$). Then by Lemma~\ref{lem:conc}, the number of $x'\in S_r$ which intersect $B_{\rho_w n}(x+e)$ for some $x\in T$ is $O(n^2)$ with probability $1-\exp(-\Omega(n^2))$. 

As before, we then also have that each $x'$ only intersects $B_{\rho_w n}(x+e)$ for $O(n^2)$ different $x\in T$, with probability $1-\exp(-\Omega(n^2))$. Thus $M_e\leq O(n^4)$, as desired. 
\end{proof}

\begin{lemma} \label{lem:internal} Fix $e\in\{0,1\}^n$ of weight at most $\rho_w n$, and let $S_r\subseteq C|_V$ have size $2^{\epsilon n/4}$. Then with probability $1-\exp(-\Omega(n^2))$ over the choice of codewords in $S_r$, there are at most $O(n^4)\cdot 2^{\epsilon n/8}$ codewords $x\in S_r$ such that  $x+e$ contains another codeword in $S_r$. 
\end{lemma}
\begin{proof} The proof of this fact appears in the full version of~\cite{myopic}, but we record it here for completeness. 

We consider arranging the elements of $S_r$ arbitrarily into a $2^{\epsilon n/8}\times 2^{\epsilon n/8}$-sized array $A$, indexed by $i,j\in [2^{\epsilon n/8}]$. 

For fixed $i$, let $r(i)=\{A(i,j)|j\in [2^{\epsilon n/8}]\}$ be the $i$th row of $A$. Similarly for fixed $j$, we can define the $j$th column of $A$ to be $c(j) = \{A(i,j)|i\in [2^{\epsilon n/8}]\}$. 
\smallskip

Define the event $E_2$ to be the event that each row $r(i)$ has $O(n^4)$ elements $x$ such that $x+e$ conflicts with some $x'\in r(i')$, for $i'\neq i$, and each column $c(j)$ similarly has $O(n^4)$ elements $x$ such that $x+e$ conflicts with some element in $c(j')$, $j'\neq j$. We claim that $E_2$ holds with probability $1-\exp(-\Omega(n^2))$ over the uniform choice of codewords in $C|_V$. 

Let us fix a row $r(i)$ (the argument for columns is identical). Note that the elements of $r(i)$ are {\em independent} of the rest of the array. Thus, by Lemma~\ref{lem:subset}, with probability $1-\exp(-\Omega(n^2))$, for any $i$, there are $O(n^4)$ elements $x\in r(i)$ such that $x+e$ conflicts with some $x'\in r(i')$, for $i'\neq i$. 

As there are $2\cdot 2^{\epsilon n/8}$ total rows and columns, we apply a union bound to conclude $E_2$ holds with probability $1-\exp(-\Omega(n^2))\cdot \exp(O(n))=1-\exp(-\Omega(n^2))$. 
\medskip

Whenever $E_2$ holds, by considering all $2^{\epsilon n/8}$ rows, we have that the number of array elements $A(i,j)$ in $A$ such that $A(i,j)+e$ conflicts with another codeword $A(i',j')$ for $i'\neq i$ is at most $O(n^4)\cdot 2^{\epsilon n/8}$. Similarly, the number of array elements $A(i,j)$ such that $A(i,j)+e$ conflicts with $A(i',j')$ for $j'\neq j$ is at most $O(n^4)\cdot 2^{\epsilon n/8}$. 

Thus, the number of codewords $A(i,j)$ such that $A(i,j)+e$ conflicts with a different codeword $A(i',j')$ is at most $O(n^4)\cdot 2^{\epsilon n/8}$ with probability $1-\exp(-\Omega(n^2))$. 
\end{proof}
\bigskip

For a fixed set $S_r\subseteq C|_V$ of size $2^{\epsilon n/4}$ and a fixed error vector $e\in \{0,1\}^n$ of weight $\rho_w n$, let $E_2(S_r,e)$ be the event that the number of codewords in $S_r$ which result in a decoding error for a fixed error vector $e$ is at most 
\[2^{3\epsilon n/16} + O(n^4)\]
Combining Lemmas~\ref{lem:external} and~\ref{lem:internal}, we see that $\Pr[E_2(S_r,e) | E_0] \geq 1-\exp(-\Omega(n^2))$. 

Recall that $\{S_i\}$ was a partition of $C|_V$ into sets of size $2^{\epsilon n/4}$, with the possible exception of the last set. Conditioned on the event $E_1$, the set $S_i$ which contains the transmitted codeword has size $2^{\epsilon n/4}$. 
In particular, we can apply a union bound over the exponentially many partition sets $S_i$ of size $2^{\epsilon n/4}$, at most $2^n$ error vectors $e$, and adversary views $V$ to conclude that $E_2(S_i,e)$ holds for all such $S_i, e$ with probability $1-\exp(-\Omega(n^2))$ over the choice of the code $C$. 

Thus, conditioned on $E_0$, the decoding error of $C$ is at most
\[\Pr[\neg E_1] + \Pr[\text{error} | E_1]\leq \Pr[\neg E_1] + \frac{O(n^4)\cdot2^{\epsilon n/8} + O(n^4)}{2^{\epsilon n/4}} = \exp\bigl(-\Omega(n)\bigr)\]
with probability $1-\exp\bigl(-\Omega(n^2)\bigr)$. 

As $E_0$ holds with probability $1-\exp(-2^{\Omega(n)})$, $C$ achieves decoding error $o(1)$ with high probability over the choice of $C$. 
\medskip

Combined with the results of Section~\ref{sec:secrecy-proof}, we have shown the following: 

\begin{thm}\label{thm:main-constr}

Let $C\subseteq\{0,1\}^n$ be a random code of rate $R=1 - h(\rho_w)-\epsilon$, and let $\Pi=\{A_i\}$ be a random partition of $C$ into subsets of size $\ell = \rho_r n - O(1)$. 

Then, with high probability over the choices of $C$ and $\Pi$, the stochastic code $C_\Pi$ which maps the $i$th message $m_i\in \{0,1\}^{Rn-\ell}$ to a random element of $A_i$ has rate
\[R(C_\Pi) \geq  1- h(\rho_r) - \rho_r-\epsilon\]
and achieves reliability and secrecy over the $(\rho_r,\rho_w)$ adversarial wiretap channel. 
\end{thm}

\section{Conclusion}
\label{sec:conc}

We have shown that the secrecy capacity of the $(\rho_r,\rho_w)$ adversarial wiretap channel is at least $1-h(\rho_w)-\rho_r$. Below we outline what we believe to be the most interesting directions for future research. 
\medskip

\noindent {\bf Exact capacity.}  The most natural open question remaining is to close the gap between Corollary~\ref{cor:upper} and Theorem~\ref{thm:main-constr}. It seems plausible that the lower bound can be improved to match the upper bound, as in the case of the Wiretap Channel II with noisy (rather than adversarial) main channel. Doing so would require a refinement of the reliability analysis to handle different input distributions.

\medskip
\noindent {\bf Explicit constructions}. The question of efficiently constructing binary codes for wiretap channels is a challenging one, and here we would be interested in any improvement over the fully random construction, including constructions which use fewer random bits. There have been constructions given for certain special cases, including when $\rho_r=0$~(\cite{complim}) and when $\rho_w=0$~(\cite{extractor}), but to the best of our knowledge, nothing is known for the general case. 

\medskip

Note that over large alphabets, both of these questions were addressed by the construction of~\cite{WSN}, which pairs folded Reed-Solomon codes, which are optimally list-decodable, with explicit ``algebraic manipulation detection'' codes. The latter ingredient is still valid over a binary alphabet, but we do not know explicit binary codes which can be list-decoded with optimal rate.

\subsubsection*{Acknowledgements}

Thanks to Vincent Tan for many helpful discussions throughout the course of this work, and for constructive suggestions which greatly improved the presentation of this work. Thanks as well to Sidharth Jaggi for sharing a draft of his work on myopic adversaries. Thanks to Omur Ozel for pointing out an error in the statement of Theorem~\ref{thm:rand-cap} in the first version of this paper. 

\bibliographystyle{alpha}

\end{document}